\documentclass[journal]{IEEEtran}
\IEEEoverridecommandlockouts

\usepackage{mypackage}

\begin{document}

\title{End-to-End Secure Connection Probability in Multi-\\Layer Networks with Heterogeneous Rician Fading 
}

\author{
Hyeonsu Lyu\IEEEauthorrefmark{1},~\IEEEmembership{Student Member,~IEEE}, Yumin Kim\IEEEauthorrefmark{1},~\IEEEmembership{Student Member,~IEEE},\newline
Hyun Jong Yang,~\IEEEmembership{Senior Member,~IEEE}
\thanks{
Hyeonsu Lyu is with the Department of Electrical Engineering, POSTECH, Korea (email: hslyu4@postech.ac.kr).
Yumin Kim and Hyun Jong Yang are with the Department of Electrical and Computer Engineering, Seoul National University, Korea (email: \{yumin0107, hjyang\}@snu.ac.kr).
Hyun Jong Yang is also with the Institute of New Media and Communications, Seoul National University, Seoul, South Korea.

\IEEEauthorrefmark{1} Hyeonsu Lyu and Yumin Kim contributed equally to this work.
}
}

\maketitle

\begin{abstract}
Ensuring physical-layer security in non-terrestrial networks (NTNs) is challenging due to their global coverage and multi-hop relaying across heterogeneous network layers, where the locations and channels of potential eavesdroppers are typically unknown.
In this work, we derive a tractable closed-form expression of the end-to-end secure connection probability (SCP) of multi-hop relay routes under heterogeneous Rician fading.
The resulting formula shares the same functional form as prior Rayleigh-based approximations but for the coefficients, thereby providing analytical support for the effectiveness of heuristic posterior coefficient calibration adopted in prior work.
Numerical experiments under various conditions show that the proposed scheme estimates the SCP with an 1\%p error in most cases; and doubles the accuracy compared with the conventional scheme even in the worst case.
As a case study, we apply the proposed framework to real-world space-air-ground-sea integrated network dataset, showing that the derived SCP accurately captures observed security trends in practical settings.
\end{abstract}

\section{Introduction}
As non-terrestrial networks (NTNs) evolve into a key enabler of future wireless systems, ensuring physical-layer security (PLS) of NTNs is becoming increasingly critical \cite{li2019physical}.
In contrast to conventional terrestrial deployments, NTNs inherently extend the communication scale from local to global, so that each transmission covers a vast footprint and thereby dramatically enlarges the attack surface \cite{Salim25-SurvTut}.
In particular, multi-hop relaying across heterogeneous network layers, such as space, air, ground, and sea layers, can expose individual links to diverse classes of potential eavesdroppers (Eves) scattered throughout the networks.
Such cross-layer vulnerabilities are fundamentally different from those in terrestrial networks, leaving significant gaps in the existing literature on PLS for multi-layer NTN architectures \cite{Guo22-SurvTut}.

PLS in multi-layer NTNs becomes challenging 
as the secure metric of a multi-hop relay chain needs to be analyzed from an end-to-end perspective \cite{MA22-TCOM}.
User data is forwarded through successive relays, so that the overall secrecy performance is jointly determined by the bottleneck among the legitimate links and the strongest eavesdropping link along the entire path.
Moreover, NTNs' extensive communication range makes it practically infeasible to accurately characterize the channels and locations of all potential Eves.
In this context, the secure connection probability (SCP) under hidden Eves has gained increasing attention in the PLS of NTNs \cite{Li25-TAero}.

\begin{table}[tb]
    \centering
    \vspace{-0.3cm}
    \caption{Comparisons of SCP Approximations}
    \label{tab:related_works}
{\setlength{\tabcolsep}{5pt}
    \begin{tabular}{cccccc}
        \toprule[.75pt]
        Ref. & Arch. & Hop & Fading & Distribution & Post-calib.\\
        \toprule[.35pt]
        \cite{Yao16-TCOM} & \textcolor{lightergray}{Ground} & Multi & Rayleigh & \textcolor{lightergray}{Exponential} & \mygraycheck \\
        \cite{lyu2025-SSIR} & NTN & \textcolor{lightergray}{Single} & Rayleigh &  \textcolor{lightergray}{Exponential} & \mygraycheck \\
        \cite{Lyu25-ICTC} & NTN & Multi & Rayleigh & \textcolor{lightergray}{Erlang} & \mygraycheck \\
        \bottomrule[.35pt]
        \addlinespace[.06cm]
        Ours & NTN & Multi & \makecell{Hetero-Rician} & Non-central-$\chi^2$ & \Xmark \\
        \bottomrule[.75pt]
    \end{tabular}
}
\vspace{-0.5cm}
\end{table}    

Research efforts to analyze the SCP of multi-hop links have been mainly reported in \cite{Yao16-TCOM, lyu2025-SSIR, Lyu25-ICTC}, but these approaches still suffer from several critical limitations.
First, \cite{Yao16-TCOM} is the first attempt to define the SCP over end-to-end relaying in terrestrial networks.
Subsequently, \cite{lyu2025-SSIR} introduces the SCP concept in NTNs, but restricts the SCP to a single-link and thus fails to capture multi-hop topologies.
\cite{Lyu25-ICTC} shows that SCP can be derived for multi-hop NTNs. 
However, these works assume Rayleigh fading for all links in the network,
which is unsuitable for multi-layer NTNs with heterogeneous base stations.
Moreover, these works conduct heuristic posterior calibration to adjust the model coefficient to compensate for their large approximation gap,
but no rigorous justification has been provided for this calibration.

\begin{figure}[!b]
    \vspace{-0.4cm}
    \centering
    \includegraphics[width=1\linewidth]{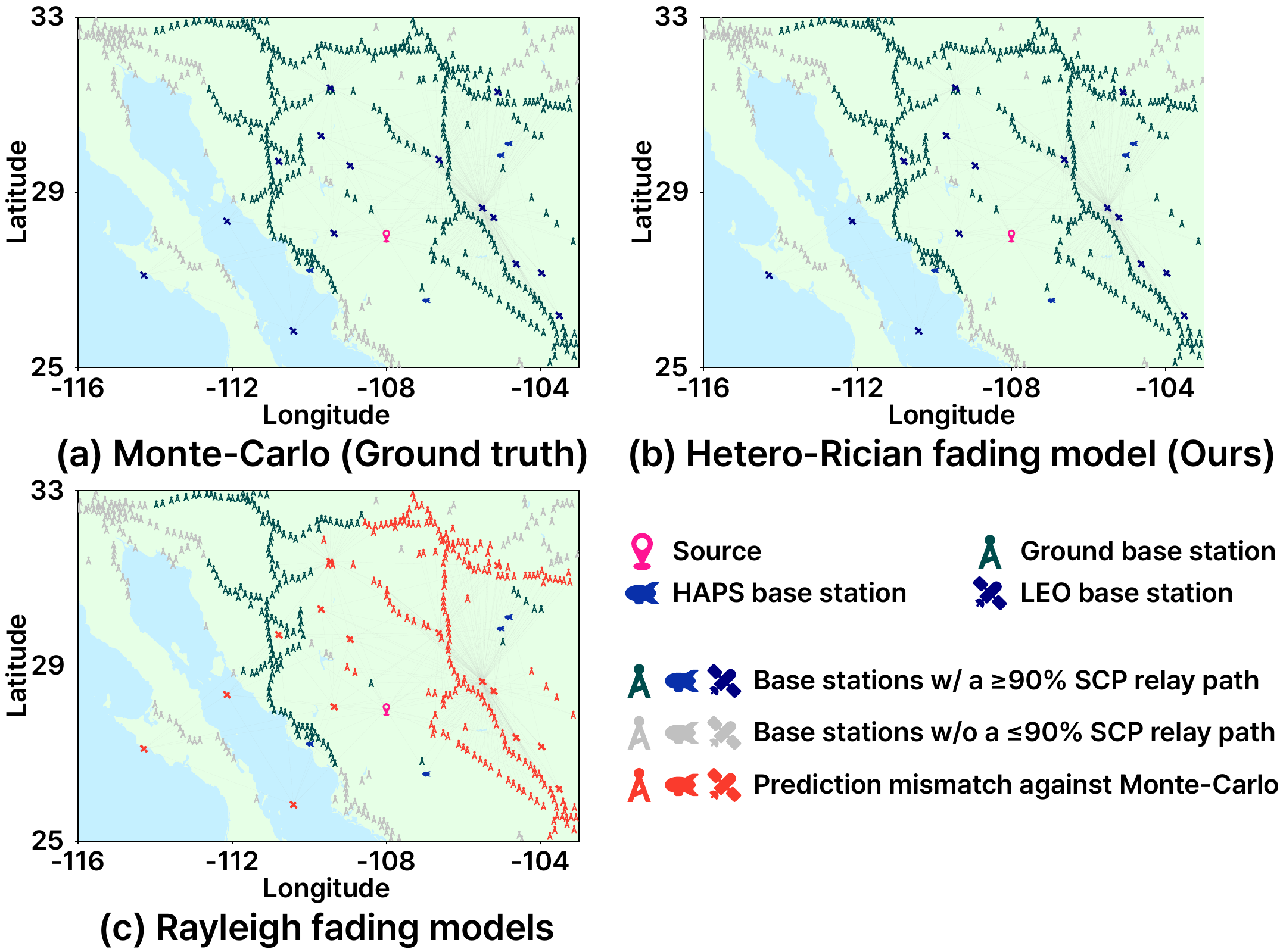}
    \vspace{-0.6cm}
    \caption{Visualization of the proposed scheme with baselines in the North America space-air-ground integrated network testbed.}
    \label{fig:visualization}
\end{figure}

This work proposes a closed-form SCP expression in the multi-layer NTNs with heterogeneous Rician channels without relying on posterior calibration when Eves are unknown.
We find that the resulting formula has the same form as the previous work \cite{Yao16-TCOM}, but with a different coefficient.
This similarity provides a theoretical explanation for why posterior calibration works effectively.
As illustrated in Fig.~\ref{fig:visualization}, the proposed model aligns much more closely with the ground truth than the previous model.
Consequently, it accurately predicts the base stations whose paths achieve an SCP exceeding 90\%, correcting the discrepancy observed in the prior work.\footnote{A detailed explanation and the additional results with an SCP threshold 99\% are provided in Appendix~\ref{section:Visualization results with a 99 SCP threshold}.}
In summary, our main contributions are \textbf{(i)} deriving a closed-form SCP expression for heterogeneous Rician channels, which was previously intractable; \textbf{(ii)} validating its accuracy in the real-world testbed; and \textbf{(iii)} establishing a theoretical foundation for the heuristic calibrations adopted by prior work.

\section{System Model and Formulation}
\subsection{System Model}
We consider a multi-hop relaying scenario in heterogeneous multi-layer NTNs and adopt decode-and-forward (DF) relaying, a standard and practical choice for NTN relay links~\cite{Bhatnagar13-CL}.
Eves are assumed to be passive, not transmitting any signals to avoid exposing their location and channel.
We thus model that Eves are independently distributed according to a homogeneous Poisson point process (HPPP) in each network layer.
This work aims to find a closed-form SCP expression using stochastic geometry in the above scenario. 

Let the index set of layers be $\mathcal{L} = \{1,\ldots,L\}$.
The sets of legitimate nodes and Eves in layer~$l$ are $\mathcal{I}_l$ and $\mathcal{M}_l$, and the overall sets are $\mathcal{I} = \bigcup_{l\in\mathcal{L}} \mathcal{I}_l$ and $\mathcal{M} = \bigcup_{l\in\mathcal{L}} \mathcal{M}_l$.
We focus on an arbitrary layer $l\in\mathcal{L}$ without loss of generality.
Suppose there are $I_l$ hops in layer~$l$; the $i$-th hop consists of the link between transmitter node~$i$ and receiver node~$i+1$ for $i \in \{0,\ldots,I_l-1\}$.
When node~$i$ transmits a signal, not only node~$i+1$ receives it, but also an Eve~$m\in\mathcal{M}_l$ may attempt to wiretap it.
Both channels experience Rician fading, and we denote the legitimate and eavesdropping channel coefficients by $h^{\mathsf{s}}_i$ and $h^{\mathsf{e}}_{(i,m)}$, respectively.
The signal-to-noise ratios (SNRs) of the legitimate and wiretap links associated with the transmission from node~$i$ are given by
\begin{align}
    \mathrm{SNR}^{\mathsf{s}}_{i} = \frac{P_i |h^{\mathsf{s}}_{i}|^2}{n_0 \big(d^{\mathsf{s}}_{i}\big)^{\alpha_i}},
    ~
    \mathrm{SNR}^{\mathsf{e}}_{(i, m)} = \frac{P_i |h^{\mathsf{e}}_{(i, m)}|^2}{n_0 \big(d^{\mathsf{e}}_{(i, m)}\big)^{\alpha_i}},
\end{align}
where $d^{\mathsf{s}}_i$ is the distance from node~$i$ to node~$i+1$, $d^{\mathsf{e}}_{(i,m)}$ is the distance from node~$i$ to Eve~$m$, $n_0$ is the noise power, $P_i$ is the transmit power of node~$i$, and $\alpha_i$ is the path-loss exponent.
Nodes in the same network layer are assumed to have the same path-loss exponent and transmit power, as they are situated in a similar transmitting terminal and wireless environment.

We define the secrecy capacity of layer~$l$ as the difference between the minimum legitimate-link rate and the maximum eavesdropping-link rate along the multi-hop path in layer~$l$, which is given by 
\begin{align}
    \hspace{-0.15cm}\xi_l \hspace{-0.05cm}=\hspace{-0.1cm}
    \Big[\hspace{-0.05cm} \log_2\hspace{-0.05cm}\big(1 + \hspace{-0.05cm}\min_{i \in \mathcal{I}_l}\hspace{-0.05cm} \mathrm{SNR}^{\mathsf{s}}_{i}\big)
    \hspace{-0.05cm}-\hspace{-0.05cm}
    \log_2\hspace{-0.05cm}\big(1 +\hspace{-0.05cm} \max_{m \in \mathcal{M}_l} \mathrm{SNR}^{\mathsf{e}}_{m}\big)\hspace{-0.05cm} \Big]^+\!\!\!,
    \label{eq:secrecy_capacity}
\end{align}
where $[x]^+=\max(x,0)$ and the combined SNR at Eve~$m$ is given by $\mathrm{SNR}^{\mathsf{e}}_m = \sum_{i \in \mathcal{I}_l} \mathrm{SNR}^{\mathsf{e}}_{(i,m)}$, which is obtained via maximal-ratio combining \cite{Brennan03-MRC}.

\subsection{Closed-form Derivation of SCP}
We define the strictly positive SCP for a multi-hop relay network across $L$ layers as the probability that all layers achieve positive secrecy capacity \cite{Jameel19-SurvTut}, which is
\begin{align}
    \mathcal{P} &=
    \mathcal{P}
    \left(
        \xi_1>0, \xi_2>0, \ldots, \xi_L>0
    \right)
    =\prod_{l\in\mathcal{L}} \mathcal{P}_l,
    \label{eq:SCP}
\end{align}
where $\mathcal{P}_l=\mathcal{P}(\min_{i \in \mathcal{I}_l}\mathrm{SNR}^{\mathsf{s}}_{i}>\max_{m\in\mathcal{M}_l}\SNR{e}_{m})$.
We now derive a closed-form expression for $\mathcal{P}_l$, which is
\begin{align}
    &\mathcal{P}_l
    =\mathbb{E}_{|h^\mathsf{s}_i|, \mathcal{M}_l}
    \Big[
        \prod_{m\in\mathcal{M}_l}\mathcal{P}\big(\epsilon_l > \SNR{e}_m \big| h^\mathsf{s}_i, \mathcal{M}_l \big)
    \Big]
    \label{eq:SCP_layer_1} \\
    &=\mathbb{E}_{|h^\mathsf{s}_i|}
    \bigg[
        \hspace{-.05cm}\exp\hspace{-.05cm}\Big[
            -\lambda_l\int_{\mathbb{R}^2}
            \mathcal{P}\big(\epsilon_l<\SNR{e}_m \big|h^\mathsf{s}_i\big)\hspace{-.1cm} ~d\bm{p}_m
        \Big]
    \bigg]
    \label{eq:SCP_layer_2} \\
    &=\mathbb{E}_{|h^\mathsf{s}_i|}
    \bigg[
        \hspace{-.05cm}\exp\hspace{-.05cm}\Big[
            \hspace{-.05cm}-\hspace{-.05cm}\lambda_l\hspace{-.1cm}\int_{\mathbb{R}^2}\hspace{-.1cm}\mathcal{P}
            \Big(\hspace{-.05cm}
                \epsilon_l<
                \sum_{i\in\mathcal{I}_l} \frac{P_i|h^{\mathsf{e}}_{(i, m)}|^2}{n_0(d^{\mathsf{e}}_m)^{\alpha_i}}
            \Big)\hspace{-.1cm} ~d\bm{p}_m
        \Big]
    \bigg] \!\!
    \label{eq:SCP_layer_3} \\
    &=\mathbb{E}_{|h^\mathsf{s}_i|}
    \bigg[
        \hspace{-.05cm}\exp\hspace{-.05cm}\Big[
            \hspace{-.05cm}-\hspace{-.05cm}\lambda_l\hspace{-.1cm}\int_{\mathbb{R}^2}\hspace{-.1cm}\mathcal{P}
            \Big(\hspace{-.05cm}
                \frac{n_0 (d^{\mathsf{e}}_m)^{\alpha_l} \epsilon_l}{P_l}\hspace{-.05cm}<\hspace{-.05cm}
                \sum_{i\in\mathcal{I}_l} |h^{\mathsf{e}}_{(i, m)}|^2
            \Big)\hspace{-.1cm} ~d\bm{p}_m
        \Big]
    \bigg],
    \label{eq:SCP_layer_4}
\end{align}
where $\epsilon_l=\min_{i\in\{0,...,I_l-1\}}\SNR{s}_{i}$ for brevity, $\lambda_l$ denotes the density of Eves in layer~$l$, and $\bm{p}_m$ denotes the location of Eve~$m$.
Equation~\eqref{eq:SCP_layer_1} follows from the independence of the HPPP, using the property
$\mathcal{P}(\max_{m\in\mathcal{M}_l} X_m < C)=\prod_{m\in\mathcal{M}_l} \mathcal{P}(X_m<C)$ for independent random variables.
Equation~\eqref{eq:SCP_layer_2} is obtained by applying the probability generating functional of the HPPP~\cite{ElSawy17-SurvTut-HPPP}, given by $\mathbb{E}_{\mathcal{M}_l}[\prod_{m\in\mathcal{M}_l}f(e_m)]=\exp \big(-\lambda_l\int_{\mathbb{R}^2}1-f(e_m) \bm{d}e_m\big)$.
Equation~\eqref{eq:SCP_layer_4} is derived by assuming uniform distances from all relays in a layer to an Eve, simplifying the distance term to $d^{\mathsf{e}}_{(i,m)} = d^{\mathsf{e}}_m$.

The distribution of the summation of channel power gains $\sum_{i\in\mathcal{I}_l} |h^{\mathsf{e}}_{(i, m)}|^2$ presents an analytical challenge, as it involves the sum of non-central chi-squared random variables with different scaling factors.
Since such a distribution does not admit a closed-form CDF, we approximate it using a gamma distribution~\cite{casella2024statistical}.
The parameters of the gamma distribution are determined through moment matching, which provides a tractable characterization while preserving the key statistical properties of the original distribution.
Let $K_i$ denote the Rician $K$-factor of hop~$i$.
Assuming unit mean channel power, $|h^{\mathsf{e}}_{(i, m)}|^2$ follows a non-central chi-squared distribution with non-centrality parameter $2K_i$ and a scaling factor $1/(2+2K_i)$.
We define $Y_l=\sum_{i\in\mathcal{I}_l} |h^{\mathsf{e}}_{(i, m)}|^2$, whose mean and variance are given by $\mathbb{E}[Y_l]=I_l$ and $\mathrm{Var} (Y_l)=\sum_{i\in\mathcal{I}_l} \frac{2K_i+1}{(K_i+1)^2}$,
where $I_l$ denotes the number of hops in layer~$l$.
The shape parameter $m_l$ and scale parameter $\theta_l$ of the equivalent gamma distribution are defined as
\begin{align}
    m_l=\frac{(\mathbb{E}[Y_l])^2}{\mathrm{Var}(Y_l)}
    ,~
    \theta_l=\frac{\mathrm{Var}(Y_l)}{\mathbb{E}[Y_l]}
\end{align}
Accordingly, by substituting the gamma CDF into~\eqref{eq:SCP_layer_4}, $\mathcal{P}_l$ can be approximated as
\begin{align}
\hspace{-.2cm}\mathcal{P}_l &\approx
    \hspace{-0.05cm}\mathbb{E}_{|h^\mathsf{s}_i|}
    \bigg[
        \hspace{-.05cm}\exp\hspace{-.05cm}\Big(
            \hspace{-.05cm}-\hspace{-.05cm}\lambda_l\hspace{-.05cm}
            \underbrace{
                \int_{\mathbb{R}^2}\hspace{-.1cm}
                1-
                \frac{\gamma\big(m_l, \frac{n_0 \epsilon_l}{P_l \theta_l} (d^{\mathsf{e}}_m)^{\alpha_l} \big)}{\Gamma(m_l)}
                ~d\bm{p}_m
            }_{J_1}
        \Big)
    \bigg]
    \\[-0.3cm]
    &=\mathbb{E}_{|h^\mathsf{s}_i|}\hspace{-0.05cm}
    \bigg[
        \hspace{-.05cm}\exp\hspace{-0.05cm}\Big(\hspace{-0.05cm}-\hspace{-0.05cm}\lambda_l C \epsilon_l^{-\frac{2}{\alpha_l}}\Big)
    \bigg],
    \label{eq:SCP_layer_5}
\end{align}
where $C=\pi \frac{\Gamma\big(m_l+\frac{2}{\alpha_l}\big)}{\Gamma(m_l)} \big(\frac{P_l \theta_l}{n_0}\big)^{\frac{2}{\alpha_l}}$.
$\Gamma(\cdot)$ and $\gamma(\cdot, \cdot)$ denote the Euler gamma function and the lower incomplete gamma function, respectively.
We then simplify $J_1$ using the lemma:

\begin{lemma}
    (Spatial Integral) The following equality holds for the integration over the two-dimensional Euclidean plane:
    \begin{equation}
        \int_{\mathbb{R}^2}\hspace{-.1cm}
                1-
                \frac{\gamma\big(m_l, k (d^{\mathsf{e}}_m)^{\alpha_l} \big)}{\Gamma(m_l)}
                ~d\bm{p}_m
        =\pi\frac{\Gamma\big(m_l+\frac{2}{\alpha_l}\big)}{\Gamma(m_l)} k^{-\frac{2}{\alpha_l}}.
        \nonumber
    \end{equation}
    \label{lemma1}
\end{lemma}
\vspace{-0.6cm}
\begin{proof}
    See Appendix~\ref{sec:Proof of SCP for Rician fading}.
\end{proof}

The expectation in \eqref{eq:SCP_layer_5} can be further expressed by applying integration by parts as
\begin{align}
    &
    \mathbb{E}_{|h^\mathsf{s}_i|}\hspace{-0.05cm}
    \Big[
        \exp\hspace{-0.05cm}\big(\hspace{-0.1cm}-\hspace{-0.1cm}\lambda_l C \epsilon_l^{-\frac{2}{\alpha_l}}\big)
    \Big]\hspace{-0.1cm}
    =\hspace{-0.1cm}\int_0^{\infty}\hspace{-0.25cm}
        \exp\hspace{-0.05cm}\Big(\hspace{-0.1cm}-\hspace{-0.1cm}\lambda_l C x^{-\frac{2}{\alpha_l}}\Big)
        f_X(x)\hspace{-0.05cm}~dx
    \nonumber\\
    &=1-\frac{2 \lambda_l C}{\alpha_l}
        \underbrace{\int_0^{\infty}\hspace{-0.2cm}
            x^{-\frac{2}{\alpha_l}-1}
            \exp\hspace{-0.05cm}\big(\hspace{-0.1cm}-\hspace{-0.1cm}\lambda_l C x^{-\frac{2}{\alpha_l}}\big)
            F_X(x)\hspace{-0.05cm}~dx
        }_{J_2},
    \label{eq:SCP_layer_6}
\end{align}
where $X=\min_{i\in\mathcal{I}_l}\SNR{s}_{i}$; $f_X(x)$ and $F_X(x)$ denote the PDF and CDF of $X$, respectively.
The CDF of $X$, the minimum of non-central chi-squared random variables, is given by
\begin{align}
    \hspace{-0.4cm}F_X(x)\hspace{-0.05cm}
    &=\hspace{-0.05cm}1\hspace{-0.05cm}-\hspace{-0.05cm}\mathcal{P}\big(\min_{i\in\mathcal{I}_l}\SNR{s}_{i}\hspace{-0.05cm}>\hspace{-0.05cm}x\big)\hspace{-0.05cm}
    =\hspace{-0.05cm}1\hspace{-0.05cm}-\hspace{-0.05cm}\prod_{i\in\mathcal{I}_l}\mathcal{P}\big(\SNR{s}_i\hspace{-0.05cm}>\hspace{-0.05cm}x\big)
    \\
    &=\hspace{-0.05cm}1\hspace{-0.05cm}-\hspace{-0.05cm}\prod_{i\in\mathcal{I}_l}
    Q_1\Big(\hspace{-0.05cm}
        (2K_i)^{\frac{1}{2}},
        \Big[2(K_i+1)\frac{n_0(d_i^{\mathsf{s}})^{\alpha_l}}{P_l}x\Big]^{\frac{1}{2}}
    \Big),\!\!
\end{align}
where $Q_1(\cdot,\cdot)$ denotes the first-order Marcum Q-function.

Since the integral in \eqref{eq:SCP_layer_6} is analytically intractable, we approximate the product of multiple Marcum Q-functions by a single equivalent Marcum Q-function:
$
    \prod_{i\in\mathcal{I}_l}Q_1(a_i,b_i\sqrt{x})
    \approx
    Q_1(\hat{a}_l,\hat{b}_l\sqrt{x}).
$
We first define $\hat{b}_l$ as $(\sum_{i\in\mathcal{I}_l}b_i^2)^{\frac{1}{2}}$, leveraging the identity $Q_1(0,b)=\exp(-b^2/2)$ \cite[Eq.~(4.45)]{simon2004digital}.
Although this identity strictly holds only when the first parameter is zero, we extend this definition to the general case of nonzero $a_i$'s.
This is justified by \cite{bocus2013approx} that the two parameters of the first-order Marcum Q-function can be treated independently with high accuracy.
Subsequently, $\hat{a}_l$ is determined via curve fitting to minimize the deviation from the exact CDF.
This approximation is empirically verified in Section~\ref{sec:marcum_validation}.
Consequently, the CDF of $X$ is approximated as
\begin{align}
    F_X(x) \approx 1-Q_1\Big(\hat{a}_l, \Big[\sum_{i\in\mathcal{I}_l}2(K_i+1)\frac{n_0(d^{\mathsf{s}}_i)^{\alpha_l}}{P_l}x\Big]^{\frac{1}{2}}\Big).
\end{align}

With this approximate CDF, the term $J_2$ in \eqref{eq:SCP_layer_6} can be evaluated by invoking the following lemma.
\begin{lemma}
    (Approximation for Marcum Q-function Integral)
    The following approximation holds for large values of $|\hat{a}_l|$:
    \begin{align}
        \hspace{-0.2cm}\int_0^{\infty}
        x^{-\frac{2}{\alpha_l}-1}
        \exp\big(\hspace{-.1cm}-\hspace{-.05cm}\lambda_l& C x^{-\frac{2}{\alpha_l}}\big)
        \big[
        1-Q_1\big(\hat{a}_l, \hat{b}_l\sqrt{x}\big)
        \big]
        ~dx
        \nonumber
        \\[-.2cm]
        &~\approx\frac{\alpha_l}{2\lambda_l C}
        \bigg[1-\exp
                \Big(\hspace{-.1cm}-\hspace{-.05cm}\lambda_l C
                    \big(\frac{\hat{b}_l^2}{\hat{a}_l^2}\big)^{\frac{2}{\alpha_l}}
                \Big)
        \bigg].
        \nonumber
    \end{align}
    \label{lemma2}
\end{lemma}
\vspace{-0.6cm}
\begin{proof}
    See Appendix~\ref{sec:Proof of SCP for Rician fading}.
\end{proof}

By substituting the result of Lemma~\ref{lemma2} into \eqref{eq:SCP_layer_6}, the closed-form of $\mathcal{P}_l$ is derived as
\begin{align}
    \mathcal{P}_l=\exp\bigg[
        -\kappa_l\Big(
            \sum_{i\in\mathcal{I}_l}
            (K_i+1)(d_i^{\mathsf{s}})^{\alpha_l}
        \Big)^{\frac{2}{\alpha_l}}
    \bigg],
    \label{eq:SCP_layer_final}
\end{align}
where $\kappa_l=
\pi\lambda_l
\frac{\Gamma(m_l+\frac{2}{\alpha_l})}{\Gamma(m_l)}
\big(\frac{2\theta_l}{\hat{a}_l^2}\big)^{\frac{2}{\alpha_l}}$.
Substituting \eqref{eq:SCP_layer_final} into $\mathcal{P}=\prod_{l\in\mathcal{L}}\mathcal{P}_l$ yields the end-to-end SCP in multi-layer network with heterogeneous Rician channels:
\begin{align}
    \mathcal{P} = \exp\Big[
        -\sum_{l\in\mathcal{L}}\kappa_l\Big(
            \sum_{i\in\mathcal{I}_l} (K_i+1)(d^\mathsf{s}_{i})^{\alpha_l}
        \Big)^{\frac{2}{\alpha_l}}
    \Big].
    \label{eq:scp_rician}
\end{align}

\section{Numerical Experiments}
\label{sec:Numerical_Experiments}
\subsection{Simulation Settings}

This section aims to validate the proposed scheme by comparing it with the previously introduced baselines in Table~\ref{tab:related_works} within a heterogeneous NTN environment.
We adopt a space, air, ground, and sea integrated network as our NTN testbed, which has been actively discussed recently.
Each network layer contains its unique wireless characteristics, including altitudes, link distances, Rician $K$-factors, and path-loss exponents which are chosen from the literature \cite{Zhang25-OJVT, sun25-Arxiv, Karapantazis05-ST}.
A full list of network parameters is listed in Table~\ref{tab:env_parameters}.

\begin{table}[htb]
    \centering
    \caption{List of Environmental parameters}
    \label{tab:env_parameters}
    {\setlength{\tabcolsep}{2pt}
    \begin{tabular}{cccc}
        \toprule[.75pt]
        Layer & Link distance (km) & Rician $\mathbf{K}_{\text{dB}}$ (mean, var) & PL exponent $\alpha_l$ \\
        \toprule[.35pt]
        LEO & $[200, 550]$ & $(13.5, 1.8)$ & $2.1$ \\
        HAPS & $[20, 380]$ & $(13.5, 1.8)$ & $2.3$ \\
        Ground & $[10, 30]$ & $(7.0, 4.0)$ & $2.9$ \\
        Sea & $[10, 30]$ & $(12.7, 1.2)$ & $2.5$ \\
        \bottomrule[.75pt]
    \end{tabular}
    }
    \vspace{-0.3cm}
\end{table}

\subsection{Validation of the Marcum Q-function Approximation}
\label{sec:marcum_validation}
This subsection validates the accuracy of the proposed approximation, which is given by
\begin{align}
    \hspace{-0.05cm}\prod_{i\in\mathcal{I}_l}
    Q_1\Big(\hspace{-0.05cm}
        (2K_i)^{\frac{1}{2}}&,
        \Big[2(K_i+1)\frac{n_0(d_i^{\mathsf{s}})^{\alpha_l}}{P_l}x\Big]^{\frac{1}{2}}
    \Big)
    \\[-0.2cm]
    &\approx Q_1\Big(\hat{a}_l, \Big[\sum_{i\in\mathcal{I}_l}2(K_i+1)\frac{n_0(d^{\mathsf{s}}_i)^{\alpha_l}}{P_l}x\Big]^{\frac{1}{2}}\Big).
    \nonumber
\end{align}
The fitting of $\hat{a}_l$ is determined by minimizing the squared error between the two functions over a logarithmically spaced grid of $x$ values ranging from $10^{-9}$ to $10^{-2}$.
We conduct 500 Monte-Carlo simulations based on the parameters in Table~\ref{tab:env_parameters}, with the number of hops uniformly selected between 2 and 7.

\begin{figure}[t]
    \centering
    \includegraphics[width=\linewidth]{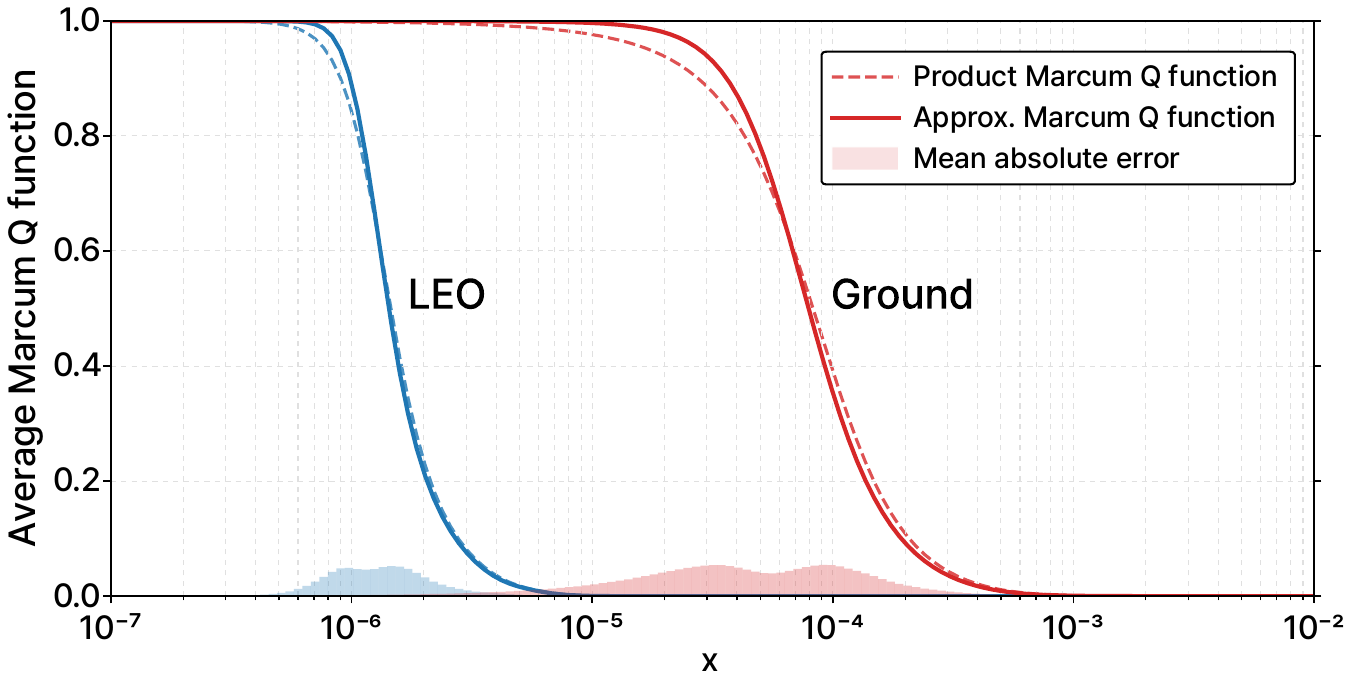}
    \caption{Comparison of the exact product of Marcum Q-functions and the proposed approximation for LEO and ground layers.}
    \vspace{-0.3cm}
    \label{fig:marcum_approx_error}
\end{figure}

Figure~\ref{fig:marcum_approx_error} presents the exact and approximate curves averaged over 500 simulations with randomly selected parameters for the LEO and ground layers, alongside the mean absolute error.
It is observed that the mean error remains below 0.06 for both layers over the entire range of $x$, ensuring negligible impact on the overall performance analysis.
Specifically, the LEO layer, which has minor variations in $K$-factors, achieves consistently higher accuracy compared to the ground layer.
Small fluctuations in $K_i$, which are associated with the first parameter $\hat{a}_l$, cause minimal distortion to the shape of the product, thereby minimizing the approximation error~\cite{bocus2013approx}.

\subsection{Performance Evaluation of the SCP}

\begin{table}[htb]
    \centering
    \vspace{-0.2cm}
    \caption{Closed-form SCP expressions per scheme}
    \label{tab:Closed-form SCP expressions per scheme}
    \vspace{-0.1cm}
{\setlength{\tabcolsep}{3pt}
    \begin{tabular}{ll}
        \toprule[.75pt]
        Method & Closed-form SCP formula\\
        \toprule[.35pt]
        Multi-hop Rayleigh \cite{Yao16-TCOM} &
        $\mathcal{P}\!=\!\exp\!\big[\!
        -\!\!\sum_{l}\textcolor[HTML]{d40b00}{\kappa_l^{\mathsf{Rayleigh}}}
        \big(
            \sum_{i} 
            \hspace{0.1mm}\raisebox{-0.5mm}{\setlength{\unitlength}{1mm}\color[HTML]{d40b00}\dashbox{1}(7.5,2.3){}\hspace{0.5mm}}
            (d^\mathsf{s}_{i})^{\alpha_l}
        \big)^{\frac{2}{\alpha_l}}
        \big]
        $
        \\
        Single-hop Rayleigh \cite{lyu2025-SSIR} &
        $\mathcal{P}\!=\!\exp\!\big[\!
        -\!\!\sum_{l}\textcolor[HTML]{d40b00}{\kappa_l^{\mathsf{Rayleigh}}}
        \big(
            \sum_{i}
            \hspace{0.1mm}\raisebox{-0.5mm}{\setlength{\unitlength}{1mm}\color[HTML]{d40b00}\dashbox{1}(7.5,2.3){}\hspace{0.5mm}}
            (d^\mathsf{s}_{i})^{\textcolor[HTML]{d40b00}{2}}
        \hspace{1.4mm}\big)
        ^{ 
            \hspace{0.1mm}\raisebox{-1mm}{\setlength{\unitlength}{1mm}\color[HTML]{d40b00}\dashbox{0.8}(3.0,2.5){}\hspace{0.3mm}}
        }
        \hspace{0.11mm}\big]
        \!\!
        $
        \\
        Multi-hop Erlang \cite{Lyu25-ICTC} &
        $\mathcal{P}\!=\!\exp\!\big[\!
        -\!\!\sum_{l}\textcolor[HTML]{d40b00}{\kappa_l^{\mathsf{Erlang}}}
        \hspace{1.82mm}
        \big(
            \sum_{i}
            \hspace{0.1mm}\raisebox{-0.5mm}{\setlength{\unitlength}{1mm}\color[HTML]{d40b00}\dashbox{1}(7.5,2.3){}\hspace{0.5mm}}
            (d^\mathsf{s}_{i})^{\alpha_l}
        \big)^{\frac{2}{\alpha_l}}
        \big]
        \!\!
        \vspace{.06cm}
        $
        \\
        \bottomrule[.35pt]
        \addlinespace[.06cm]
        Hetero-Rician (\textbf{Ours}) & 
        $\mathcal{P}\!=\!\exp\!\big[\!
        -\!\!\sum_{l}\kappa_l
        \hspace{7mm}
        \big(
            \sum_{i}
            \hspace{-0.3mm}(K_i\hspace{-1.0mm}+\hspace{-1.0mm}1)
            (d^\mathsf{s}_{i})^{\alpha_l}
        \big)^{\frac{2}{\alpha_l}}
        \big]
        \!\!
        $
        \\[.06cm]
        \bottomrule[.75pt]
    \end{tabular}
    \par\nointerlineskip\vskip 1mm
    {
    \raggedright
        * \textcolor[HTML]{d40b00}{Red highlights} indicate deviation from the proposed formula. \\
    * $\kappa_l^{\mathsf{Rayleigh}} = \pi\lambda_l\Gamma(1+\frac{2}{\alpha_l})\Gamma(1-\frac{2}{\alpha_l})$,~
    $\kappa_l^{\mathsf{Erlang}} = \pi \lambda_l \frac{\Gamma(1-\frac{2}{\alpha_l}) \Gamma(I_l+\frac{2}{\alpha_l})}{\Gamma(I_l)}$
    }
} 
\vspace{-0.3cm}
\end{table}

The proposed scheme \eqref{eq:scp_rician} is evaluated against three baseline schemes, whose closed-form expressions are provided in Table~\ref{tab:Closed-form SCP expressions per scheme}.\footnote{We include the derivation of Erlang distribution in Appendix~\ref{sec:Proof of SCP for Rayleigh fading} as no direct closed-form formula is provided in \cite{Lyu25-ICTC}.}
The comparison schemes have a mathematical form highly similar to our scheme, except for the coefficient terms. 
The following experiments (Figs.\ref{fig:Kfactor_vs_scp}-\ref{fig:density_vs_scp}) show that the proposed scheme provides more accurate results than the baseline schemes.
This explains why the heuristic calibration adopted in the baselines to adjust the coefficients can achieve high accuracy.

We additionally implement Monte-Carlo simulation, as defined in \eqref{eq:SCP}, that computes SCP by deploying Eves according to an HPPP, calculating the $\xi_l$ for each layer, and computing the empirical probability of the event $\xi_l>0$ for all $l\in\mathcal{L}$ over 100,000 simulations.
We observe how accurately each scheme predicts the SCP for a given environment, using Monte-Carlo simulation as the benchmark.

Figure \ref{fig:Kfactor_vs_scp} shows the SCP across various Rician $K$ factors and Eve densities in the isolated terrestrial and LEO layers.
This setup eliminates inter-layer variability and clearly reveals how the Rician channel characteristics affect the SCP behavior of the proposed scheme.
The link distances are randomly generated based on the ranges in Table~\ref{tab:env_parameters}, while the Eve density and the Rician $K$-factor are fixed for each scenario.
As the Rician $K$-factor increases from 0 to 14~dB, the SCP monotonically increases, indicating that a stronger line-of-sight component benefits the legitimate links more than the eavesdropping links.
The proposed scheme demonstrates a high degree of agreement with the Monte-Carlo simulations across a wide range of parameters, whereas the comparison schemes appear as horizontal lines since they are restricted to Rayleigh fading and cannot capture the effect of varying Rician $K$-factors.
We observe up to 10\%p gap in the low-$K$ regime between the proposed scheme and Monte-Carlo simulations 
as the approximation of Lemma~\ref{lemma2} becomes more accurate as the Rician $K$-factor increases.\footnote{A larger $K_i$ yields a larger effective shift in $\hat{a}_l$ to estimate the product of Marcum $Q$-functions.}
Nevertheless, the proposed analysis remains highly accurate in realistic scenarios as practical Rician $K$-factors typically lie in the range of 10–20~dB~\cite{kim2024cell}.

\begin{figure}[t]
    \centering
    \includegraphics[width=\linewidth]{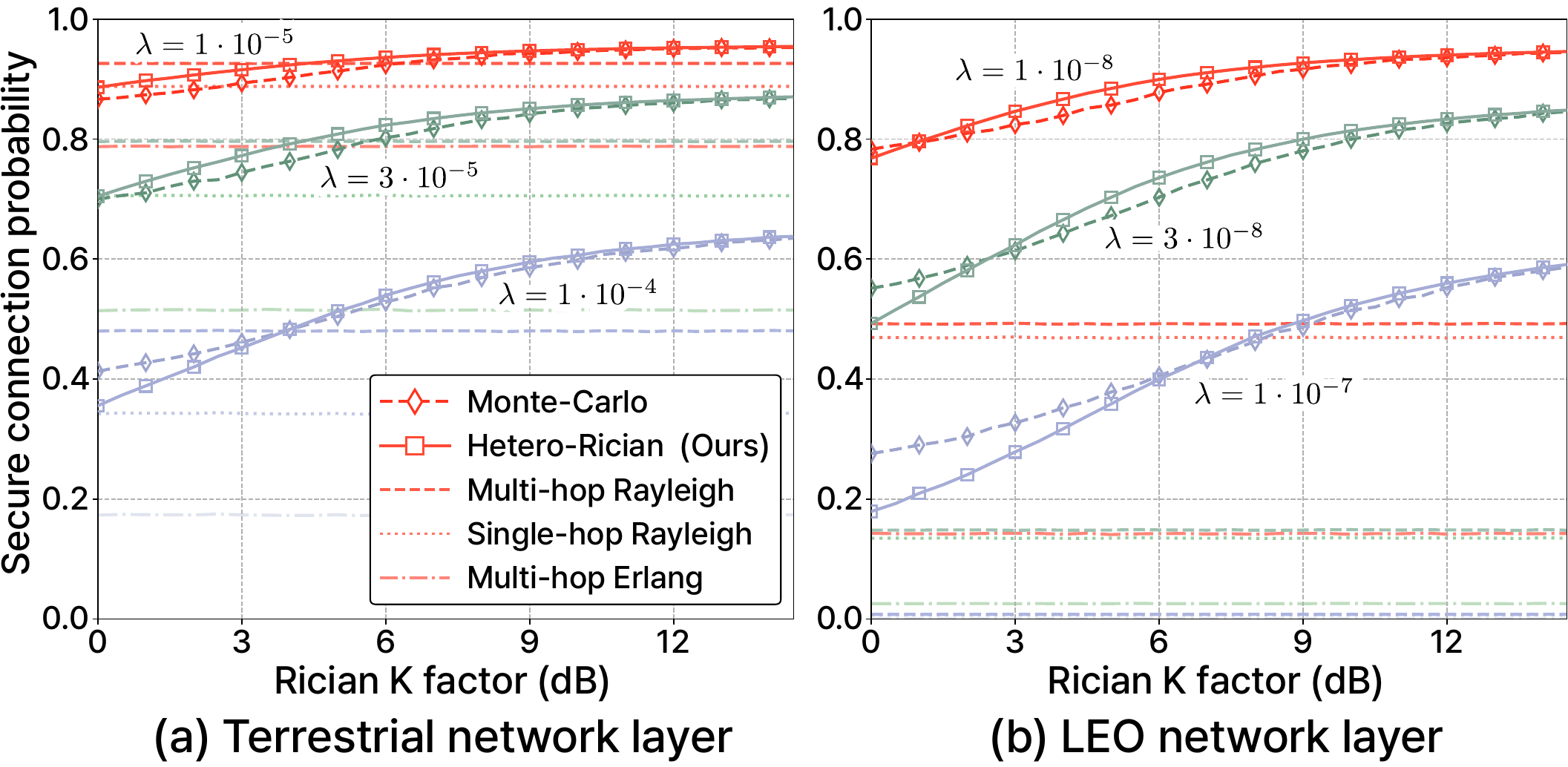}
    \caption{Impact of Rician $K$-factor on SCP in terrestrial and LEO layers. The red, green, and blue curves correspond to Eve densities of $10^{-5}$, $3 \times 10^{-5}$, and $10^{-4}$ (terrestrial), and $10^{-8}$, $3 \times 10^{-8}$, and $10^{-7}$ (LEO).}
    \vspace{-0.3cm}
    \label{fig:Kfactor_vs_scp}
\end{figure}

\begin{figure}[htb]
    \centering    \includegraphics[width=0.75\linewidth]{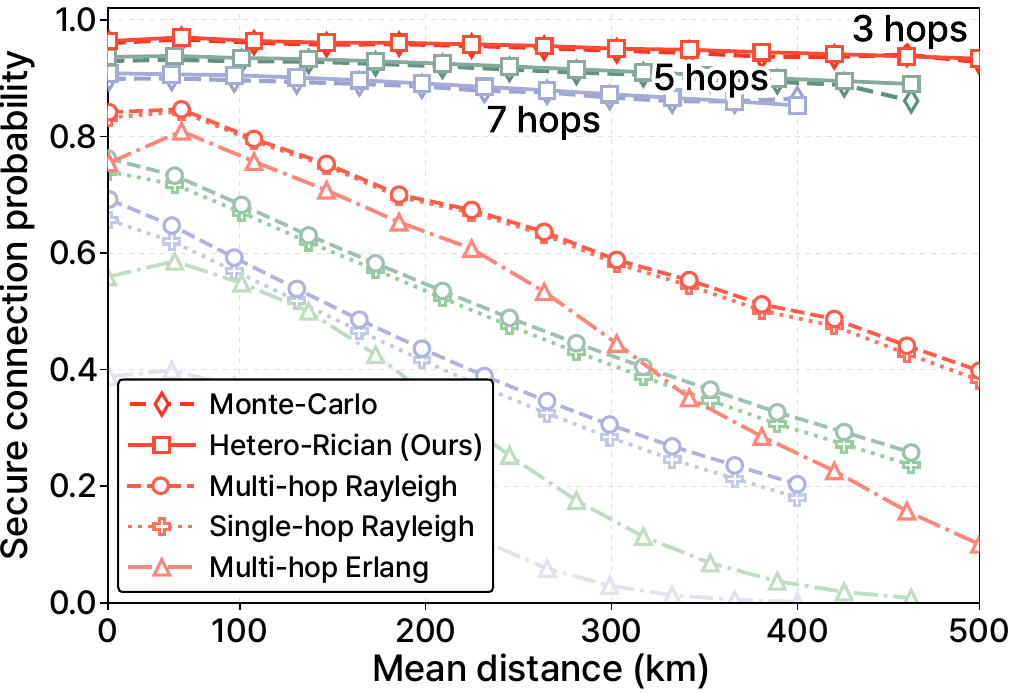}
    \caption{SCP of \{3,5,7\} hops over various average link distances. 
    The red, green, and blue curves correspond to the \{3,5,7\} hop scenarios, respectively.}
    \vspace{-0.3cm}
    \label{fig:distance_vs_scp}
\end{figure}

Figure \ref{fig:distance_vs_scp} shows the SCP across various hops and average link distances in the four-layer NTN.
In our random-topology simulations, the effective distance range in Fig.~\ref{fig:distance_vs_scp} decreases as the number of hops increases.
This is because the number of samples that form paths approaching the maximum feasible connection distance decreases as the hop count increases.
The proposed scheme maintains a high SCP that matches the Monte-Carlo results, whereas the Rayleigh-based schemes experience rapid degradation.
This observation is consistent with Fig.~\ref{fig:visualization}, where the comparison scheme fails to identify the secure connections in most cases, while our scheme successfully captures them.

Figure~\ref{fig:density_vs_scp} examines the SCP behavior versus the Eve densities in a four-layer NTN scenario.
Specifically, the hop count is randomly selected over two to seven, and each hop is randomly assigned to one of the four layers, thereby constructing a more comprehensive experimental setup than those in Figs.~\ref{fig:Kfactor_vs_scp} and~\ref{fig:distance_vs_scp}.
The SCP expression shows nonlinear behavior as the Eve density varies geometrically from $10^{-9}$ to $10^{-4}$.
The conventional schemes assuming Rayleigh fading fail to capture the heterogeneity in Rician fading.
Thus, the modeling errors accumulate as the Eve density increases, leading to a significantly steeper SCP decrease.
In contrast, the proposed hetero-Rician scheme acts as a tight upper bound that remains very close to the Monte-Carlo results over the entire range.
Numerically, the proposed scheme achieves a mean absolute error of approximately 0.034, which is about $4.5\times$ smaller than that of the best Rayleigh-based baseline.
In addition, in the high-SCP regime with $\mathcal{P}\geq90\%$, the average gap between the proposed scheme and the Monte-Carlo curve is below 1\%p, confirming that the proposed method provides highly accurate predictions in the primary region of interest for system design~\cite{wang2016physical}.

\begin{figure}[t]
    \centering
    \includegraphics[width=0.75\linewidth]{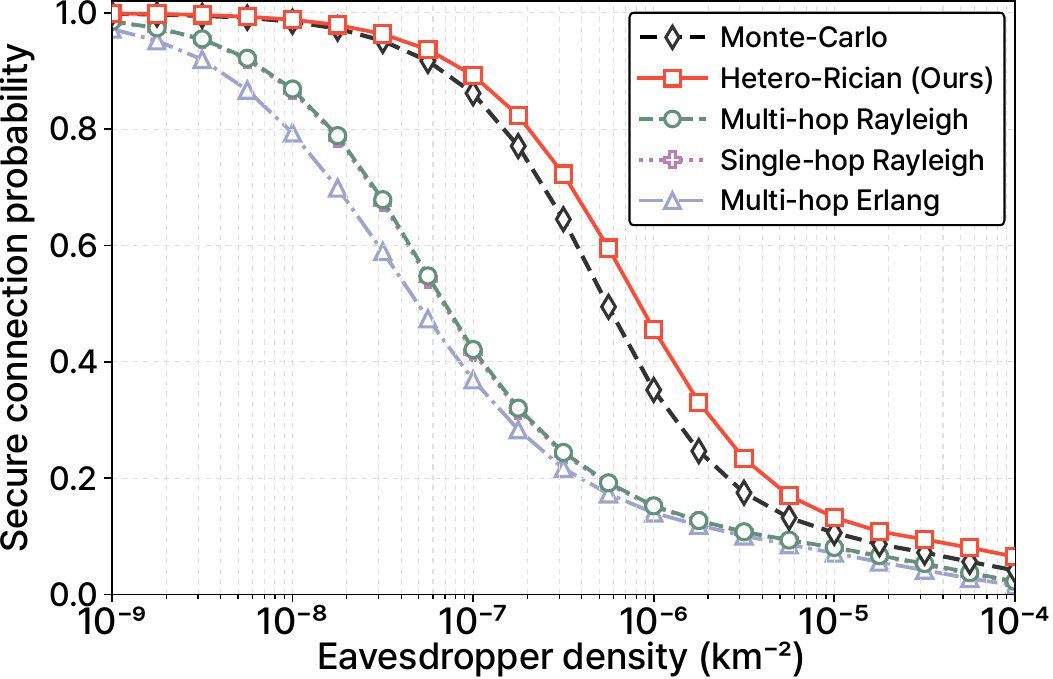}
    \caption{SCP under various Eve densities.}
    \vspace{-0.3cm}
    \label{fig:density_vs_scp}
\end{figure}

\section{Conclusion and Discussion}
In this paper, we analytically characterized the end-to-end SCP of multi-hop routes in heterogeneous multi-layer NTNs under heterogeneous Rician fading.
Unlike prior studies restricted to Rayleigh channels or single-hop scenarios,
this work provides an analytical derivation of the multi-hop SCP under heterogeneous Rician channels, which has remained largely unexplored.
The derived closed-form SCP expression closely matches Monte-Carlo simulations across diverse network conditions.
In particular, the SCP prediction remains robust under diverse Rician $K$-factors, Eve densities, and link distances, even in heterogeneous multi-layer environments.

Despite its generality and accuracy, this study has several limitations that motivate promising avenues for future work.
First, our analysis is restricted to decode-and-forward relaying; extending the framework to amplify-and-forward, compress-and-forward, or hybrid relaying would be a valuable contribution.
Second, this work does not incorporate adversarial or cooperative jamming, which is critical in practical NTN operations \cite{Mu25-IOT}.
Integrating spatially distributed jammers and their power/placement strategies into the SCP analysis remains an important direction.
We hope this study serves as a groundwork for future research on PLS in NTNs.

\bibliographystyle{IEEEtran}
\bibliography{references.bib}

@STRING{IEEE_J_AES        = "{IEEE} Trans. Aerosp. Electron. Syst."}

@STRING{IEEE_J_IOT        = "{IEEE} Internet Things J."}

@STRING{IEEE_J_COML       = "{IEEE} Commun. Lett."}

@STRING{IEEE_J_COM        = "{IEEE} Trans. Commun."}

@STRING{IEEE_S_COM        = "{IEEE} Commun. Surveys Tuts."}

@STRING{IEEE_OJ_VT         = "{IEEE} Open J. Veh. Technol."}

@STRING{IEEE_J_PROC       = "Proc. {IEEE}"}

@ARTICLE{Guo22-SurvTut,
  author={Guo, Hongzhi and Li, Jingyi and Liu, Jiajia and Tian, Na and Kato, Nei},
  journal=IEEE_S_COM, 
  title={A Survey on Space-Air-Ground-Sea Integrated Network Security in 6{G}}, 
  year={2022},
  volume={24},
  number={1},
  pages={53-87},
  doi={10.1109/COMST.2021.3131332}}

@ARTICLE{Jameel19-SurvTut,
  author={Jameel, Furqan and Wyne, Shurjeel and Kaddoum, Georges and Duong, Trung Q.},
  journal=IEEE_S_COM,
  title={A Comprehensive Survey on Cooperative Relaying and Jamming Strategies for Physical Layer Security}, 
  year={2019},
  volume={21},
  number={3},
  pages={2734-2771},
  doi={10.1109/COMST.2018.2865607}}

@ARTICLE{Yao16-TCOM,
  author={Yao, Jianping and Feng, Suili and Zhou, Xiangyun and Liu, Yuan},
  journal=IEEE_J_COM,
  title={Secure Routing in Multihop Wireless Ad-Hoc Networks With Decode-and-Forward Relaying}, 
  year={2016},
  volume={64},
  number={2},
  pages={753-764},
  doi={10.1109/TCOMM.2015.2514094}}

@misc{lyu2025-SSIR,
      title={Secure Multi-Hop Relaying in Large-Scale Space-Air-Ground-Sea Integrated Networks}, 
      author={Hyeonsu Lyu and Hyeonho Noh and Hyun Jong Yang and Kaushik Chowdhury},
      year={2025},
      eprint={2505.00573},
      archivePrefix={arXiv},
      primaryClass={eess.SP},
}

@ARTICLE{bocus2013approx,
  author={Bocus, Mohammud Z. and Dettmann, Carl P. and Coon, Justin P.},
  journal=IEEE_J_COML,
  title={An Approximation of the First Order {M}arcum {Q}-Function with Application to Network Connectivity Analysis}, 
  year={2013},
  volume={17},
  number={3},
  pages={499-502},
  doi={10.1109/LCOMM.2013.011513.122462}}

@INPROCEEDINGS{Lyu25-ICTC,
  author={Lyu, Hyeonsu and Yang, Hyun Jong},
  booktitle={2025 Int. Conf. ICT Converg. (ICTC)}, 
  title={Maneuver by Airstreams for Stratospheric Balloon Base Stations}, 
  year={2025},
  volume={},
  number={},
  pages={},
  doi={}}

@ARTICLE{Zhang25-OJVT,
  author={Zhang, Chao and Li, Qingchao and Xu, Chao and Yang, Lie-Liang and Hanzo, Lajos},
  journal=IEEE_OJ_VT,
  title={Space-Air-Ground Integrated Networks: Their Channel Model and Performance Analysis}, 
  year={2025},
  volume={6},
  number={},
  pages={1501-1523},
  doi={10.1109/OJVT.2025.3575360}}

@misc{sun25-Arxiv,
      title={Modeling and Analysis of Land-to-Ship Maritime Wireless Channels at 5.8 {GHz}}, 
      author={Shu Sun and Yulu Guo and Meixia Tao and Wei Feng and Jun Chen and Ruifeng Gao and Ye Li and Jue Wang and Theodore S. Rappaport},
      year={2025},
      eprint={2507.15969},
      archivePrefix={arXiv},
      primaryClass={eess.SP},
      url={https://arxiv.org/abs/2507.15969}, 
}

@ARTICLE{Karapantazis05-ST,
  author={Karapantazis, S. and Pavlidou, F.},
  journal=IEEE_S_COM,
  title={Broadband communications via high-altitude platforms: A survey}, 
  year={2005},
  volume={7},
  number={1},
  pages={2-31},
  doi={10.1109/COMST.2005.1423332}}

@book{abramowitz1965handbook,
  title={Handbook of mathematical functions: with formulas, graphs, and mathematical tables},
  author={Abramowitz, Milton and Stegun, Irene A},
  volume={55},
  year={1965},
  publisher={Courier Corporation}
}

@article{kim2024cell,
  title={Cell-free massive non-terrestrial networks},
  author={Kim, Seungnyun and Wu, Jiao and Shim, Byonghyo and Win, Moe Z},
  journal={IEEE Journal on Selected Areas in Communications},
  year={2024},
  publisher={IEEE}
}

@ARTICLE{Li25-TAero,
  author={Li, Xin and Li, Yongjun and Song, Xinkang and Li, Jianjia and Zhao, Shanghong},
  journal=IEEE_J_AES,
  title={Physical Layer Security Analysis of Hybrid {FSO/RF} Systems in {SAGIN} With Channel Imperfection and Interference}, 
  year={2025},
  volume={61},
  number={5},
  pages={13156-13171},
  doi={10.1109/TAES.2025.3576315}}

@ARTICLE{Brennan03-MRC,
  author={Brennan, D.G.},
  journal=IEEE_J_PROC,
  title={Linear diversity combining techniques}, 
  year={2003},
  volume={91},
  number={2},
  pages={331-356},
  doi={10.1109/JPROC.2002.808163}}

@book{casella2024statistical,
  title={Statistical inference},
  author={Casella, George and Berger, Roger},
  year={2024},
  publisher={Chapman and Hall/CRC}
}

@ARTICLE{ElSawy17-SurvTut-HPPP,
  author={ElSawy, Hesham and Sultan-Salem, Ahmed and Alouini, Mohamed-Slim and Win, Moe Z.},
  journal=IEEE_S_COM,
  title={Modeling and Analysis of Cellular Networks Using Stochastic Geometry: A Tutorial}, 
  year={2017},
  volume={19},
  number={1},
  pages={167-203},
  doi={10.1109/COMST.2016.2624939}}

@ARTICLE{Bhatnagar13-CL,
  author={Bhatnagar, Manav R.},
  journal=IEEE_J_COML, 
  title={On the Capacity of Decode-and-Forward Relaying over {Rician} Fading Channels}, 
  year={2013},
  volume={17},
  number={6},
  pages={1100-1103},
  doi={10.1109/LCOMM.2013.050313.122813}}

@ARTICLE{Mu25-IOT,
  author={Mu, Shunkai and Lei, Hongjiang and Park, Ki-Hong and Pan, Gaofeng},
  journal=IEEE_J_IOT,
  title={Finite Block-Length Covert Communication in Space-Air-Ground Integrated Networks}, 
  year={2025},
  volume={},
  number={},
  pages={1-1},
  doi={10.1109/JIOT.2025.3531881}}

@book{simon2004digital,
  title={Digital communication over fading channels},
  author={Simon, Marvin K and Alouini, Mohamed-Slim},
  year={2004},
  publisher={John Wiley \& Sons}
}

@article{li2019physical,
  title={Physical-layer security in space information networks: A survey},
  author={Li, Bin and Fei, Zesong and Zhou, Caiqiu and Zhang, Yan},
  journal={IEEE Internet of things journal},
  volume={7},
  number={1},
  pages={33--52},
  year={2019},
  publisher={IEEE}
}

@ARTICLE{Salim25-SurvTut,
  author={Salim, Sara and Moustafa, Nour and Reisslein, Martin},
  journal=IEEE_S_COM, 
  title={Cybersecurity of Satellite Communications Systems: A Comprehensive Survey of the Space, Ground, and Links Segments}, 
  year={2025},
  volume={27},
  number={1},
  pages={372-425},
  doi={10.1109/COMST.2024.3408277}}

@article{wang2016physical,
  title={Physical layer security in heterogeneous cellular networks},
  author={Wang, Hui-Ming and Zheng, Tong-Xing and Yuan, Jinhong and Towsley, Don and Lee, Moon Ho},
  journal=IEEE_J_COM,
  volume={64},
  number={3},
  pages={1204--1219},
  year={2016},
  publisher={IEEE}
}

@ARTICLE{MA22-TCOM,
  author={Ma, Yuanyuan and Lv, Tiejun and Pan, Gaofeng and Chen, Yunfei and Alouini, Mohamed-Slim},
  journal=IEEE_J_COM,
  title={On Secure Uplink Transmission in Hybrid {RF}-{FSO} Cooperative Satellite-Aerial-Terrestrial Networks}, 
  year={2022},
  volume={70},
  number={12},
  pages={8244-8257},
  doi={10.1109/TCOMM.2022.3214891}}

\newpage

\begin{appendices}

\section{Visualization results with a 99\% SCP threshold}
\label{section:Visualization results with a 99 SCP threshold}
Figure~\ref{fig:visualization_scp99} shows the nodes that have relay paths whose SCP is at least 99\% from the source.
We use the real-world SAGSIN dataset in \cite{lyu2025-SSIR} to deploy space, HAPS, and ground base stations in southwest America.
we connect adjacent base stations using the environment parameters in \ref{tab:env_parameters}.
Eve densities are set to $10^{-8}$ across all layers.
As shown in Fig.~\ref{fig:visualization_scp99}b, the proposed scheme predicts an SCP of at least 99\% for two satellite base stations and two ground base stations, whereas the Monte Carlo scheme predicts their SCP to be below 99\%.
This is consistent with the results in Figs.~\ref{fig:Kfactor_vs_scp}-\ref{fig:density_vs_scp}, where the proposed scheme behaves like an upper bound on the Monte Carlo scheme.
In contrast, under the Rayleigh fading model, the actual SCP is severely underestimated; consequently, except for a few base stations, it predicts an SCP below 0.99 for almost all base stations, leading to large errors.

\begin{figure}[htb]
    \centering
    \includegraphics[width=1\linewidth]{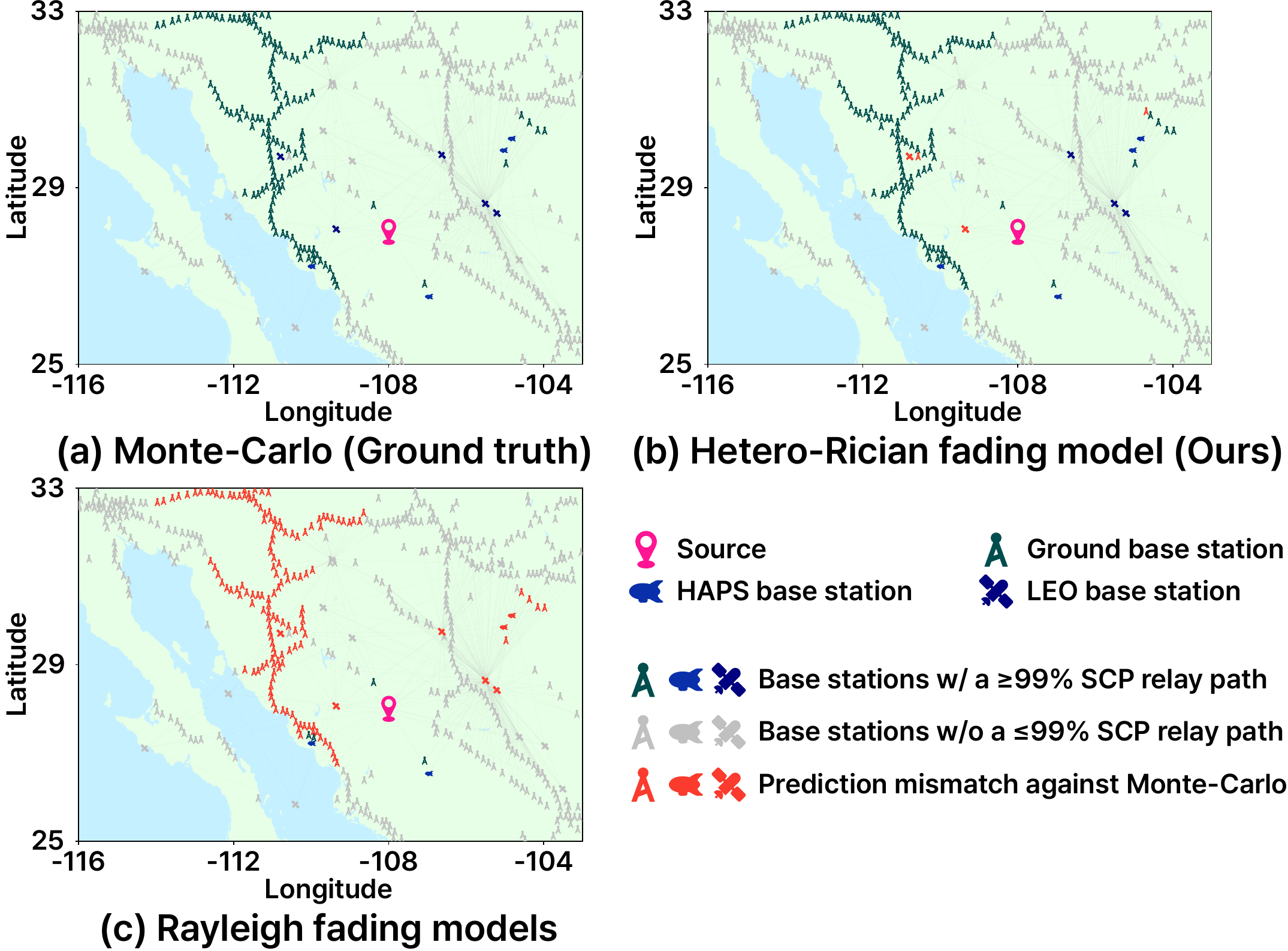}
    \caption{Visualization of the proposed scheme with baselines in the North America SAGIN dataset.}
    \label{fig:visualization_scp99}
\end{figure}
\section{Proof of SCP for Rician fading}
\label{sec:Proof of SCP for Rician fading}
\setcounter{lemma}{0}
\begin{lemma}
    (Spatial Integral) The following equality holds for the integration over the two-dimensional Euclidean plane:
    \begin{equation}
        \int_{\mathbb{R}^2}\hspace{-.1cm}1-
                \frac{\gamma\big(m_l, k(d^{\mathsf{e}}_m)^{\alpha_l}\big)}{\Gamma(m_l)} ~d\bm{p}_m
        =\pi\frac{\Gamma(m_l+\frac{2}{\alpha_l})}{\Gamma(m_l)} k^{-\frac{2}{\alpha_l}}.
    \end{equation}
\end{lemma}
\begin{proof}
\begin{align}
    &\int_{\mathbb{R}^2}\hspace{-.1cm}1-
                \frac{\gamma\big(m_l, k(d^{\mathsf{e}}_m)^{\alpha_l}\big)}{\Gamma(m_l)} ~d\bm{p}_m
    \\
    &=\int_0^{\infty}\int_0^{2\pi}\bigg[1-\frac{\gamma\big(m_l,kr^{\alpha_l}\big)}{\Gamma(m_l)}\bigg]r~d\theta dr
    \\
    &=\frac{2\pi}{\alpha_l}
                \int_0^{\infty}
                    s^{\frac{2}{\alpha_l}-1}
                    \bigg[
                        1-
                        \frac{\gamma(m_l, ks)}{\Gamma(m_l)}
                    \bigg]
                ~ds
    \\
    &=\frac{2\pi}{\alpha_l} \int_0^{\infty}
        s^{\frac{2}{\alpha_l}-1}
        \frac{1}{\Gamma(m_l)}
        \exp(-ks) (ks)^{m_l}
        \nonumber \\
        &\quad\times
        \int_0^{\infty}
            \exp(-ksu)
            (1+u)^{m_l-1}
        ~du~ds
    \\
    &=\frac{2\pi}{\alpha_l} \frac{1}{\Gamma(m_l)}
        k^{m_l}
        \int_0^{\infty} (1+u)^{m_l-1}
        \nonumber \\
        &\quad\times
        \int_0^{\infty}
            s^{m_l+\frac{2}{\alpha_l}-1}
            \exp\big(-k(1+u)s\big)
        ~dsdu
    \\
    &=\frac{2\pi}{\alpha_l} \frac{\Gamma\big(m_l+\frac{2}{\alpha_l}\big)}{\Gamma(m_l)}
        k^{-\frac{2}{\alpha_l}}
        \int_0^{\infty} (u+1)^{-1-\frac{2}{\alpha_l}}~du
    \\
    &=\pi
    \frac{\Gamma\big(m_l+\frac{2}{\alpha_l}\big)}{\Gamma(m_l)}
    k^{-\frac{2}{\alpha_l}}.
\end{align}
\end{proof}

\begin{lemma}
    (Approximation for Marcum Q-function Integral)
    The following approximation holds for large values of $|\hat{a}_l|$:
    \begin{align}
        \hspace{-0.2cm}\int_0^{\infty}
        x^{-\frac{2}{\alpha_l}-1}
        \exp\big(\hspace{-.1cm}-\hspace{-.05cm}\lambda_l& C x^{-\frac{2}{\alpha_l}}\big)
        \big[
        1-Q_1\big(\hat{a}_l, \hat{b}_l\sqrt{x}\big)
        \big]
        ~dx
        \\[-.2cm]
        &~\approx\frac{\alpha_l}{2\lambda_l C}
        \bigg[1-\exp
                \Big(\hspace{-.1cm}-\hspace{-.05cm}\lambda_l C
                    \big(\frac{\hat{b}_l^2}{\hat{a}_l^2}\big)^{\frac{2}{\alpha_l}}
                \Big)
        \bigg].
        \nonumber
    \end{align}
\end{lemma}
\begin{proof}
\begin{align}
    &\int_0^{\infty}
        \hspace{-.2cm}x^{-\frac{2}{\alpha_l}-1}\hspace{-.05cm}
        \exp\hspace{-.05cm}\Big(\hspace{-.1cm}-\hspace{-.05cm}\lambda_l C x^{-\frac{2}{\alpha_l}}\Big)\hspace{-.05cm}
        \Big[
        1\hspace{-.05cm}-\hspace{-.05cm}Q_1\big(\hat{a}_l, \hat{b}_l\sqrt{x}\big)\hspace{-.05cm}
        \Big]\hspace{-.1cm}
    ~dx
    \\
    &=\int_0^{\infty}
        \hspace{-.2cm}x^{-\frac{2}{\alpha_l}-1}
        \sum_{n=0}^{\infty}\frac{1}{n!} \Big(\hspace{-.1cm}-\hspace{-.05cm}\lambda_l C x^{-\frac{2}{\alpha_l}}\Big)^n\hspace{-.1cm}
        \exp\Big(\hspace{-.1cm}-\hspace{-.05cm}\frac{\hat{a}_l^2}{2}\Big)
        \nonumber \\
        &\quad\times
        \sum_{k=0}^{\infty}
            \frac{\gamma\big(1+k, \frac{\hat{b}_l^2x}{2}\big)}{k!\Gamma(1+k)}
            \Big(\frac{\hat{a}_l^2}{2}\Big)^k
        ~dx
    \\
    &=\exp\Big(\hspace{-.1cm}-\hspace{-.05cm}\frac{\hat{a}_l^2}{2}\Big)
        \sum_{n=0}^{\infty}\frac{1}{n!} (-\lambda_l C)^n
        \sum_{k=0}^{\infty}\sum_{m=0}^{\infty}
            \frac{1}{k!}
            \Big(\frac{\hat{a}_l^2}{2}\Big)^k
            \\
            &\quad\times
            \frac{\big(\frac{\hat{b}_l^2}{2}\big)^{1+k+m}}{\Gamma(k+m+2)}
            \hspace{-.1cm}\int_0^{\infty}\hspace{-.2cm}
                x^{-\frac{2}{\alpha_l}(n+1)+k+m}
                \exp\Big(\hspace{-.1cm}-\hspace{-.05cm}\frac{\hat{b}_l^2x}{2}\Big)
            ~dx
    \nonumber \\
    &=\exp\Big(\hspace{-.1cm}-\hspace{-.05cm}\frac{\hat{a}_l^2}{2}\Big)
        \Big(\frac{\hat{b}_l^2}{2}\Big)^{\frac{2}{\alpha_l}}
        \sum_{n=0}^{\infty}
            \frac{1}{n!}
            \bigg(\hspace{-.1cm}-\hspace{-.05cm}\lambda_l C \Big(\frac{\hat{b}_l^2}{2}\Big)^{\frac{2}{\alpha_l}}\bigg)^n
        \nonumber \\
        &\quad\times
        \sum_{k=0}^{\infty}
            \frac{1}{k!}
            \Big(\frac{\hat{a}_l^2}{2}\Big)^k
            \sum_{m=0}^{\infty}
                \frac{\Gamma\big(k+m+1-\frac{2}{\alpha_l}(n+1)\big)}{\Gamma(k+m+2)}
    \\
    &=\frac{\alpha_l}{2}
        \Big(\frac{\hat{b}_l^2}{2}\Big)^{\frac{2}{\alpha_l}}
        \sum_{n=0}^{\infty}
            \frac{1}{n!}
            \bigg(\hspace{-.1cm}-\hspace{-.05cm}\lambda_l C \Big(\frac{\hat{b}_l^2}{2}\Big)^{\frac{2}{\alpha_l}}\bigg)^n
            \nonumber \\
            &\quad\times
            \Gamma\Big(1-\frac{2}{\alpha_l}(n+1)\Big)
            M\Big(\frac{2}{\alpha_l}(n+1),1,-\frac{\hat{a}_l^2}{2}\Big)
    \\
    &=\frac{\alpha_l}{2\lambda_l C}
        \Bigg[1-
            \sum_{n=0}^{\infty}
                \frac{1}{n!}
                \bigg(\hspace{-.1cm}-\hspace{-.05cm}\lambda_l C \Big(\frac{\hat{b}_l^2}{2}\Big)^{\frac{2}{\alpha_l}}\bigg)^n
                \nonumber \\
                &\quad\times
                \Gamma\Big(1-\frac{2n}{\alpha_l}\Big)
                M\Big(\frac{2n}{\alpha_l},1,-\frac{\hat{a}_l^2}{2}\Big)
        \Bigg]
    \\
    &\approx \frac{\alpha_l}{2\lambda_l C}
        \Bigg[1-
            \sum_{n=0}^{\infty}
                \frac{1}{n!}
                \bigg(\hspace{-.1cm}-\hspace{-.05cm} \lambda_l C
                    \Big(\frac{\hat{a}_l^2}{2}\Big)^{-\frac{2}{\alpha_l}}
                    \Big(\frac{\hat{b}_l^2}{2}\Big)^{\frac{2}{\alpha_l}}
                \big)^n
        \Bigg]
    \label{eq:lemma2_approx}
    \\
    &=\frac{\alpha_l}{2\lambda_l C}
        \Bigg[1-\exp
                \bigg(\hspace{-.1cm}-\hspace{-.05cm}\lambda_l C
                    \Big(\frac{\hat{b}_l^2}{\hat{a}_l^2}\Big)^{\frac{2}{\alpha_l}}
                \bigg)
        \Bigg],
\end{align}
where $M(\cdot,\cdot,\cdot)$ is the confluent hypergeometric function.
Using the asymptotic expansion $M(a,b,z) \approx \frac{\Gamma(b)}{\Gamma(b-a)} (-z)^{-a}+\frac{\Gamma(b)}{\Gamma(a)}e^zz^{a-b}$ for large $|z|$ from \cite[Eq.~(13.5.1)]{abramowitz1965handbook}, and noting that the second term vanishes for $\Re(z)<0$, we obtain \eqref{eq:lemma2_approx}.
\end{proof}

\section{Proof of SCP for Rayleigh fading}
\label{sec:Proof of SCP for Rayleigh fading}
We derive the SCP probability for the Rayleigh channel.
With a property that a sum of independent and identically distributed exponential random variables follows an Erlang distribution, \eqref{eq:SCP_layer_3} is derived as
\begin{align}
    &\hspace{-.15cm}\mathbb{E}_{|h^\mathsf{s}_i|}\hspace{-.05cm}
    \Bigg[\hspace{-.05cm}
        \exp\hspace{-.05cm}\bigg[
            \hspace{-.1cm}-\hspace{-.05cm}\lambda_l\hspace{-.1cm}\int_{\mathbb{R}^2}\hspace{-.15cm}\mathcal{P}
            \bigg(\hspace{-.05cm}
                \frac{n_0(d^{\mathsf{e}}_m)^{\alpha_l}}{P_l} \epsilon_l
                \hspace{-.05cm}<\hspace{-.05cm}
                \sum_{i\in\mathcal{I}_l} |h^{\mathsf{e}}_{(i,m)}|^2\hspace{-.05cm}
            \bigg)\hspace{-.1cm} ~d\bm{p}_m
        \bigg]\hspace{-.1cm}
    \Bigg]
    \label{eq:SCP_layer_rayleigh_1}
    \\
    &\hspace{-.25cm}=\hspace{-.05cm}\mathbb{E}_{|h^\mathsf{s}_i|}\hspace{-.05cm}
    \Bigg[\hspace{-.05cm}
        \exp\hspace{-.05cm}\bigg[
            \hspace{-.1cm}-\hspace{-.05cm}\lambda_l\hspace{-.1cm}\int_{\mathbb{R}^2}\hspace{-.15cm}\mathcal{P}
            \bigg(\hspace{-.05cm}
                \sum_{i\in\mathcal{I}_l} \frac{1}{i!} \exp(-K \epsilon_l) (K \epsilon_l)^i\hspace{-.05cm}
            \bigg)\hspace{-.1cm} ~d\bm{p}_m
        \bigg]\hspace{-.1cm}
    \Bigg]
    \label{eq:SCP_layer_rayleigh_2}
    \\
    &\hspace{-.25cm}\geq
    \exp
    \Bigg[\hspace{-0.05cm}
        -\hspace{-0.05cm}\lambda_l \hspace{-0.1cm}
        \int_{\mathbb{R}^2}\hspace{-0.05cm}
        \mathbb{E}_{|h^\mathsf{s}_i|}
        \bigg[
            \sum_{i\in\mathcal{I}_l} \frac{1}{i!} \exp(-K \epsilon_l) (K \epsilon_l)^i
        \bigg]
        ~d\bm{p}_m
    \Bigg]
    \label{eq:SCP_layer_rayleigh_3}
    \\
    &\hspace{-.25cm}=
    \exp \hspace{-0.05cm}
    \Bigg[ \hspace{-0.05cm}
        -\hspace{-0.05cm}\lambda_l \hspace{-0.1cm}
        \int_{\mathbb{R}^2} \hspace{-0.05cm}
            1\hspace{-0.05cm}-\hspace{-0.05cm}
            \bigg(
                \frac{(d^{\mathsf{e}}_m)^{\alpha_l}}
                    {(d^{\mathsf{e}}_m)^{\alpha_l}
                    \hspace{-0.05cm}+\hspace{-0.05cm}
                    \sum_{i\in\mathcal{I}_l} \hspace{-0.05cm}
                    (d^\mathsf{s}_{i})^{\alpha_l}}
            \bigg)^{I_l} \hspace{-0.15cm}
        ~d\bm{p}_m
    \Bigg]
    \label{eq:SCP_layer_rayleigh_4}
    \\
    &\hspace{-.25cm}=
        \exp \hspace{-0.05cm}
        \bigg[ \hspace{-0.05cm}
           -\kappa_l
           \Big(
                \sum_{i\in\mathcal{I}_l} (d_{i}^{\mathsf{s}})^{\alpha_l}
           \Big)^{\frac{2}{\alpha_l}}
        \bigg],
    \label{eq:SCP_layer_rayleigh_5}
\end{align}
where $K=\frac{n_0(d^{\mathsf{e}}_m)^{\alpha_l}}{P_l}$ and $\kappa_l = \pi \lambda_l \frac{\Gamma\big(1-\frac{2}{\alpha_l}\big) \Gamma\big(I_l+\frac{2}{\alpha_l}\big)}{\Gamma(I_l)}$.
\eqref{eq:SCP_layer_rayleigh_3} is obtained by applying Jensen's inequality.
\eqref{eq:SCP_layer_rayleigh_4} is derived from Lemma~\ref{lemma3}, and \eqref{eq:SCP_layer_rayleigh_5} is from \cite[Eq.~(26)]{Yao16-TCOM}.

Substituting \eqref{eq:SCP_layer_rayleigh_5} into $\mathcal{P}=\prod_{l\in\mathcal{L}}\mathcal{P}_l$ gives the SCP probability for Rayleigh channel as
\begin{align}
    \mathcal{P}=\exp\bigg[
        -\sum_{l\in\mathcal{L}}\kappa_l\Big(
            \sum_{i\in\mathcal{I}_l} (d^\mathsf{s}_{i})^{\alpha_l}
        \Big)^{\frac{2}{\alpha_l}}
    \bigg],
\end{align}
where $\kappa_l = \pi \lambda_l \frac{\Gamma\big(1-\frac{2}{\alpha_l}\big) \Gamma\big(I_l+\frac{2}{\alpha_l}\big)}{\Gamma(I_l)}$.

\begin{lemma}
    Let $X=\min_{i=1,...,N} X_i$ for independent exponential random variables $X_i\sim\exp(\lambda_i)$ for $i=1,...,N$.
    Then,
    \begin{equation}
        \hspace{-.15cm}\mathbb{E}_{X}\hspace{-.05cm}
        \bigg[
            \sum_{i=0}^{N-1}\frac{1}{i!} \exp(-C X) (C X)^i
        \bigg]
        \hspace{-.05cm}= 1\hspace{-.05cm}-\hspace{-.05cm}
        \bigg(\hspace{-.05cm}
            \frac{C}{C+\sum_{i=1}^{N}\lambda_i}\hspace{-.05cm}
        \bigg)^{\hspace{-.1cm}N}\!.
    \label{eq:lemma3}
    \end{equation}
    \label{lemma3}
\end{lemma}
\begin{proof}
    Noting that $X$ is exponentially distributed with rate $\lambda=\sum_{i=1}^N \lambda_i$, we can derive \eqref{eq:lemma3} as follows:
    \begin{align}
        &\mathbb{E}_{X}
        \Big[
            \sum_{i=0}^{N-1}\frac{1}{i!} \exp(-C X) (C X)^i
        \Big]
        \\
        &\hspace{-.1cm}=\int_0^{\infty} \sum_{i=0}^{N-1}\frac{1}{i!} \exp(-C x) (C x)^i \lambda \exp(-\lambda x)~dx
        \\
        &\hspace{-.1cm}=\lambda \sum_{i=0}^{N-1} \frac{C^i}{i!} \int_0^{\infty} x^i \exp\hspace{-.05cm}\big(\hspace{-.05cm}-\hspace{-.05cm}(C+\lambda)x\big)~dx
        \\
        &\hspace{-.1cm}=\lambda \sum_{i=0}^{N-1} \frac{C^i}{i!} \cdot \frac{\Gamma(i+1)}{(C+\lambda)^{i+1}}
        \\
        &\hspace{-.1cm}=\frac{\lambda}{C+\lambda} \sum_{i=0}^{N-1} \Big( \frac{C}{C+\lambda} \Big)^i
        \\
        &\hspace{-.1cm}=1- \Big( \frac{C}{C+\lambda} \Big)^{N}.
    \end{align}
    As a result, plugging in $N=I_l$, $C=\frac{n_0(d^{\mathsf{e}}_m)^{\alpha_l}}{P_l}$ and $\lambda=\sum_{i\in\mathcal{I}_l} \frac{n_0(d^\mathsf{s}_{i})^{\alpha_i}}{P_l}$ leads to \eqref{eq:SCP_layer_rayleigh_4}.
\end{proof}

\end{appendices}

\end{document}